\documentclass[a4paper,12t,reqno]{article}
\usepackage[utf8]{inputenc}

\title{Methods for Collisions in Some \\ Algebraic Hash Functions}

\author{Simran Tinani\thanks{This research is supported by armasuisse Science and Technology.}}

\usepackage{graphicx}
\usepackage{amsmath, amssymb}
\usepackage[utf8]{inputenc}
\usepackage{subcaption}
\usepackage{tabularx}
\usepackage{amsmath,amssymb,amsthm}
\usepackage[english]{babel}
\usepackage{mathtools}
\usepackage{algorithmic}
\usepackage{enumerate}
\usepackage{hyperref}
\usepackage[numbers]{natbib}
\usepackage[dvipsnames]{xcolor}
\usepackage{color}
\usepackage{comment}

\usepackage{hyperref}
\hypersetup{
  colorlinks   = true, 
  urlcolor     = blue, 
  linkcolor    = OliveGreen, 
  citecolor   = red 
}

\usepackage{tabularx}
\usepackage{booktabs}
\usepackage{makecell}
\usepackage{longtable}
\usepackage[ruled,vlined]{algorithm2e}
\usepackage{algorithmic}
\usepackage{listings}
\usepackage{etoolbox}
\usepackage{lipsum}
\usepackage[left =2.7cm, right = 2.7 cm]{geometry}

\usepackage{url}

\usepackage{amssymb}
\usepackage{amsfonts}
\usepackage{amsmath}
\usepackage{comment}
\usepackage{tabularx}
\usepackage{booktabs}
\usepackage{makecell}
\usepackage{longtable}
\usepackage{graphicx}
\usepackage{caption}
\usepackage{lipsum}
\usepackage{xcolor}
\usepackage{parskip}
\usepackage{caption}
\usepackage{floatrow}
\usepackage[ruled,vlined]{algorithm2e}
\usepackage{algorithmic}
\usepackage{hyperref, enumitem}

\newcommand{\Z}{\mathbb{Z}}

\newcommand{\F}{\mathbb{F}}

\newfloatcommand{capbtabbox}{table}[][\FBwidth]

\newcommand{\fp}{\mathbb{F}_p}
\newcommand{\f}{\mathbb{F}}
\newcommand{\fpk}{{\mathbb{F}_p}^k}

\newtheorem{theorem}{Theorem}

\newtheorem{problem}{Problem}

\newtheorem{lemma}{Lemma}
\newtheorem{proposition}{Proposition}
\newtheorem{corollary}{Corollary}
\theoremstyle{definition}

\newtheorem{definition}{Definition}
\newtheorem{example}{Example}
\theoremstyle{remark}

\makeatletter
\newenvironment{pf}[1][\proofname]{\par
  \pushQED{\qed}%
  \normalfont \topsep0\p@\relax
  \trivlist
  \item[\hskip\labelsep\itshape
  #1\@addpunct{.}]\ignorespaces
}{%
  \popQED\endtrivlist\@endpefalse
}
\makeatother

\begin{document}
\maketitle

\begin{abstract} This paper focuses on devising methods for producing collisions in algebraic hash functions that may be seen as generalized forms of the well-known Z\'emor and Tillich-Z\'emor hash functions. In contrast to some of the previous approaches, we attempt to construct collisions in a structured and deterministic manner by constructing messages with triangular or diagonal hashes messages.  Our method thus provides an alternate deterministic approach to the method for finding triangular hashes in \cite{hardeasy}.
We also consider the generalized Tillich-Z\'emor hash functions over $\fp^k$ for $p\neq 2$, relating the generator matrices to a polynomial recurrence relation, and derive a closed form for any arbitrary power of the generators. We then provide conditions for collisions, and a method to maliciously design the system so as to facilitate easy collisions, in terms of this polynomial recurrence relation.
Our general conclusion is that it is very difficult in practice to achieve the theoretical collision conditions efficiently, in both the generalized Z\'emor and the generalized Tillich-Z\'emor cases. Therefore, although the techniques are interesting theoretically, in practice the collision-resistance of the generalized Z\'emor functions is reinforced. 
\end{abstract}

\section{Introduction}
Let $\mathcal{A}$ be an alphabet and $\mathcal{A}^*$ denote the set of all finite-length words in $\mathcal{A}$ and $\mathcal{A}^n$ denote the set of all words up to length $n$ in $\mathcal{A}$. A length $n$ hash function, or compression function, is a map $\mathcal{A}^*\rightarrow \mathcal{A}^n$ which takes messages of arbitrary length to fixed-length message digests. A hash function $h:\mathcal{A}^*\rightarrow \mathcal{A}^n$ is called a cryptographic hash function if it satisfies the following properties:

\begin{itemize}
\item Collision-resistance: it is computationally infeasible to find a pair $x, x'$ of distinct messages such that $h(x)=h(x')$.
\item Second pre-image resistance: given a message $x$, it is computationally infeasible to find another message $x'\neq x$ such that $h(x)=h(x')$.
\item One-wayness: given a hash value $y\in \mathcal{A}^n$ it is computationally infeasible to find a pre-image $x\in \mathcal{A}$ such that $h(x)=y$.
\end{itemize}

Cryptographic hash functions are also often required to exhibit the avalanche effect, under which a small modification in the message text causes a big change in the hash. This prevents the hash value from leaking information about the message string, and ensures no visible correlation between the hashes of related strings. A hash function that does not exhibit this property is called mallaeble.

Several widely used hash functions, including the NIST-standardized SHA (Secure Hash Algorithms) functions \cite{Penard2008OnTS}, are built from block ciphers. While the state of the art block cipher-based hash functions have shown considerable resiliency to attacks, their security is nevertheless heuristic; in other words, it does not reduce to a well-known difficult mathematical problem. The search for a provably secure hash function is therefore, at the least, of theoretical interest.

The idea of building hash functions from a group and its corresponding Cayley graph was introduced in \cite{Zemor1991hash}. 
The generic design of Cayley hash functions has several advantages over traditional hash functions. Firstly, their security is equivalent to some concise mathematical problem. Cayley hash functions are also inherently parallelizable, i.e.\ allow for simultaneous computation of the hash value of different parts of the message, and recombining these at the end. 
While the Tillich-Z\'emor hash functions are slower than SHA-256, they can be designed to have reasonably efficient implementations: in \cite{effic} it is stated that they can be made faster than SHA-1. In particular, when fields of characteristic 2 are used, the group law is the most efficient. 

There may also be other disadvantages of Cayley hash functions. By themselves, such hash functions are inherently mallaeble: given a hash $h(m)$ of an unknown message $m$, $h(x_1|| m || x_2)= h(x_1) h(m) h(x_2)$ for any texts $x_1, x_2$. Further, they lack preimage resistance for small messages. However, in \cite{effic}, the authors describe a heuristic modification which seemingly resolves these issues, as well as enhances the efficiency. Cayley hash functions have a simple, elegant, clear design related to hard mathematical problems. A very generic description is given below.

\begin{definition}[Cayley hash function]
   Let $G$ be a finite group with a set of generators $S$ and $\mathcal{A}$ be an alphabet
the same size as $S$. Given an injective map $\pi : \mathcal{A} \rightarrow S$, one may define the hash value of the message $x_1x_2 \dots x_k$ to be the group element $\pi(x_1)\pi(x_2) \ldots \pi(x_k)$.
\end{definition}

Typically, one is concerned only with binary messages, and therefore with two-generator Cayley hash functions. The common design principle behind the Cayley hash functions is the performance of a walk on a regular Cayley graph according to the bits of an input message, the last vertex giving the hash value. 

  \begin{definition}[Cayley Graph] Let $G$ be a group and $S$ a subset of elements of $G$. The Cayley graph of $C_{G,S}=(V,E)$ of $G$ with respect to $S$ is defined to have vertices $v_g$ corresponding to each element $g \in G$, and edges $(v_{g_1}, v_{g_2})\in E \iff \exists s\in S$ such that $g_2 =g_1 s$. Here, the set $S$ is the set of graph generators. \end{definition}

  The security of a Cayley hash function defined on a group $G$ is determined by the properties of the corresponding graph. More precisely, the difficulty of producing collisions for such a hash function depends on the difficulty of solving the \textit{factorization problem} in a certain finite non-abelian group.

    \begin{definition}[Factorization problem]
        Let $G$ be a group with generators $S=\{s_1, \ldots, s_k\}$ and $L>0$ be a fixed constant. Given $g\in G$, return $m_1,\ldots, m_L$ and $\ell \leq L$, with $m_i\in \{1, \ldots, k\}$ such that $\prod\limits_{i=1}^{\ell}s_{m_i}=g$. The factorization problem for $g=1$ fixed is called the representation problem. 
    \end{definition}


Apart from hash functions, the difficulty of factorization in finite matrix groups has been used to build some key exchange, encryption and authentication schemes \cite{ironwood, walnut,kayawood}.

Factorization in finite groups is related to the famous conjecture of Babai on the diameter of Cayley graphs, which states that diameter of any undirected Cayley graph of non-abelian simple group is polylogarithmic in group size. In other words, ``short" factoriations always exist for finite non-abelian groups, regardless of the choice of generating set. The conjecture is known to be true for certain groups like $SL_2(2,\fp), SL_2(2,\f_{2^k})$. However, existing generic proofs for any generating set are non-constructive. 
In \cite{origtz}, the following Cayley hash function was defined in the group $G=SL_2(\fp)$:

\begin{definition}[Z\'emor Hash Function]
    Choose the generators $A_0=\begin{pmatrix}
        1 & 1\\ 0 & 1
    \end{pmatrix}$, $A_1= \begin{pmatrix}
        1 & 0\\ 1 & 1
    \end{pmatrix}$ of $SL_2(\fp)$. For a message $m=m_1m_2 \ldots m_k \in \{0,1\}^*$ define $H(m_1\ldots m_k)=A_{m_1}\ldots A_{m_k}$.
\end{definition}
This hash function, given in \cite{origtz} was attacked for both collisions and preimages in \cite{Tillich1993GrouptheoreticHF} using the Euclidean algorithm. %
However, this attack is specific to the generators $A_0$ and $A_1$. In \cite{rubiks}, it is claimed that the system is potentially secure by replacing $A_0$ and $A_1$ by $A_0^2$ and $A_1^2$. We define the \emph{generalized Z\'emor hash function} as above, but using generators $A_0=\begin{pmatrix} 1 & \alpha
 \\ 0 & 1 \end{pmatrix}$ and $A_1=\begin{pmatrix} 1 & 0
 \\ \beta & 1 \end{pmatrix}$ over $\fpk$. While the generators $A_0$ and $A_1$ have multiplicative order $p$, and so one trivially has collisions of length $p$ with the empty word, we nevertheless also consider the case $k>1$ and study the difficulty of finding non-trivial collisions (i.e.\ collisions of non-empty words).

 The following variant of the above hash function was later proposed in \cite{origtz2}.
 
\begin{definition}[Tillich-Z\'emor Hash function]
     Let $n>0$ and $p(x)$ be an irreducible polynomial over $\F_2$. Write $K=\F_2[x]/p(x)$. Consider $A_0=\begin{pmatrix} x & 1 \\ 1 & 0
     \end{pmatrix}$ and $A_1=\begin{pmatrix} x & x+1 \\ 1 & 1
     \end{pmatrix}$, which are generators of $SL_2(K)$. For a message $m=m_1m_2 \ldots m_k \in \{0,1\}^*$ define $H(m_1\ldots m_k)=A_{m_1}\ldots A_{m_k} \pmod{p(x)} $. $K=\f_2[x]/\langle q(x)\rangle \cong F_{2^n}$.
 \end{definition}

Collisions for this hash function for these set of generators were found in \cite{grassl}. This attack uses the structure in hash values of palindromic messages, and a result of Mesirov-Sweet  \cite{Mesirov1987ContinuedFE} on the Euclidean algorithm for polynomials $x$ and $x+1$ in characteristic 2.  This attack was extended to recover preimages in \cite{preimgs}. Since current attacks do not allow controlling of the forms of collisions, and work specifically for the set of generators $A_0$ and $A_1$, the security is an open problem for general parameters. 

In \cite{towards}, the authors provide a new heuristic algorithm for factoring generic generator sets by reducing them to contain a ``trapdoor" matrix. Their algorithm is subexponential in time, memory, and factorization lengths.  This algorithm, however, does not yield a practical attack on the Tillich-Z\'emor design with generic generators for $n=160$. In fact, to the best of knowledge till date, the collision and preimage resistance are recovered by replacing $A_0$ and $A_1$ by $B_0=\begin{pmatrix} x^2 & 1 \\ 1 & 0
     \end{pmatrix}$ and $B_1=\begin{pmatrix} x+1 & 1 \\ 1 & 0
     \end{pmatrix}$ \cite{hardeasy}. 
We define the \emph{generalized Tillich-Z\'emor hash function} as above, but using generators $A_0=\begin{pmatrix}\alpha & 1 \\ 1 & 0 \end{pmatrix}$ and $A_1=\begin{pmatrix}\beta & 1 \\ 1 & 0 \end{pmatrix}$ where $\alpha, \beta \in \fpk$. 

 It has also been proposed to use LPS Ramanujan graphs \cite{lps} to construct Cayley hash functions. However, these have also been cryptanalyzed \cite{lpscryptanalysis}. Some other methods of cryptanalysis for Cayley hash functions are discussed in \cite{rubiks}.

\paragraph{Collisions from Triangular and Diagonal Matrices}
In \cite{hardeasy}, it is shown that if one can produce ``sufficiently many" messages whose images in the matrix groups are upper/lower triangular, then one can find collisions of the generalized Z\'emor and Tillich-Z\'emor hash functions. 

 \begin{proposition}\label{triang}
Let $n$ be such that discrete logarithms can be solved in $\F_{2^n}^*$. Let
$\mathcal{D}$, $\mathcal{T}^{up}$, $\mathcal{T}^{low}$, $\mathcal{L}^{v}$, $\mathcal{R}^{v}\subset SL_2(\F_{2^n})$ be the subgroups of diagonal, upper and lower triangular matrices and the subgroup of matrices with left or right eigenvector $v$. If an attacker can compute $N$ random elements $M_i$ of one of these subgroups together with bit sequences $m_i$ of length at most $L$ hashing to these matrices, then he can also find a message $m$ such that $H_{ZT}(m) = I$. The message $m$ has expected size smaller than $N L2^{n/N}$ in the diagonal case and smaller than $N L2^{1+n/N}$ in the other cases.
\end{proposition}

In this same work, the authors use random probabilistic search to find pre-images of upper/lower triangular matrices, and subsequently uses these to produce collisions.  This approach works for any Cayley hash function, and for $SL_2$ groups over any finite field.

\section*{Summary of Contributions}
 
This paper focuses on devising methods for producing collisions in algebraic hash functions that may be seen as generalized forms of the well-known Z\'emor and Tillich-Z\'emor hash functions. In contrast to some of the previous approaches, we attempt to construct collisions in a structured and deterministic manner.

In Section~\ref{gen-zem}, we introduce a method for constructing messages with triangular or diagonal hashes messages. For this, we extend existing hash values in $SL_2(\fpk)$ into triangular or diagonal form by multiplying with products of the form $A_0^mA_1^n$, where $A_0=\begin{pmatrix} 1 & \alpha
 \\ 0 & 1 \end{pmatrix}$ and $A_1=\begin{pmatrix} 1 & 0
 \\ \beta & 1 \end{pmatrix}$ denote generators. More precisely, we consider the following problem.

\begin{problem}Given a matrix $C\in SL_2(\fpk)$ formed as product of $A_0$ and $A_1$, find the conditions under which there exist integers $m$ and $n$ (of size significantly smaller than $p^k$) such that $CA_0^m A_1^n$ is upper/lower triangular, or even diagonal. Compute $m$ and $n$ if they exist.
\end{problem}

 We also discuss the application of this method to produce collisions, and the feasibility and efficiency thereof. Our method thus provides an alternate deterministic approach to the method for finding triangular hashes in \cite{hardeasy}.

 In Section~\ref{gen-TZ}, we consider the generalized Tillich-Z\'emor hash functions over $\fpk$ for $p\neq 2$, relating the generator matrices to a polynomial recurrence relation, and derive a closed form for any arbitrary power of the generators. We then provide conditions for collisions, and a method to maliciously design the system so as to facilitate easy collisions, in terms of this polynomial recurrence relation.
 
On simplifying the general criteria, and through experiments, our general conclusion is that it is very difficult in practice to achieve the theoretical collision conditions efficiently, in both the generalized Z\'emor and the generalized Tillich-Z\'emor cases. Therefore, although the techniques are interesting theoretically, in practice the collision-resistance of the generalized Z\'emor functions is reinforced.

\section{Generalized Z\'emor Hash functions} \label{gen-zem}

\begin{definition}[Generalized Z\'emor hash functions] Consider the generators $A_0=\begin{pmatrix} 1 & \alpha
 \\ 0 & 1 \end{pmatrix}$ and $A_1=\begin{pmatrix} 1 & 0
 \\ \beta & 1 \end{pmatrix}$ in the group $SL_2(\fpk)$. For a message $m=m_1m_2 \ldots m_k \in \{0,1\}^*$ define the hash value $H(m_1\ldots m_k)=A_{m_1}\ldots A_{m_k}$.
\end{definition}

 The Z\'emor hash function proposed in \cite{origtz} is a special case of this, with $\alpha= 1=\beta$. In \cite{rubiks} it was suggested that security is preserved when $s_0^2$ and $s_1^2$ are used instead, these cases correspond to the values $\alpha=2=\beta$.
Thus, the generalized Z\'emor hash functions are so far secure.

Clearly, $A_0^m=\begin{pmatrix}1 & m\alpha \\ 0 & 1 \end{pmatrix}$ and $A_1^n=\begin{pmatrix}1 & 0 \\ n\beta & 1 \end{pmatrix}$ for any integers $m$ and $n$.
$A_0$ and $A_1$ have multiplicative orders $p$, so trivial collisions of length $p$ with the empty word always exist. Nevertheless, we study the general case of non-trivial collisions in the finite field $\fpk$.  In general, one is interested in finding collisions with a significantly smaller length than the orders, at most, say $\mathcal{O}(\sqrt{p})$.

\subsection{Euclidean Algorithm Attack for $\alpha= \beta=1$}
In this subsection, we describe the attack of \cite{Tillich1993GrouptheoreticHF}, which uses the Euclidean algorithm, and explain why it fails if one chooses $\alpha\neq 1$, $\beta\neq 1$ in the hash function design. 
Consider a matrix $X=\begin{psmallmatrix}a & b\\ c & d\end{psmallmatrix}\in SL_2(\fp)$ and suppose that we have found a non-identity matrix
  $Y=\begin{psmallmatrix}A & B\\ C & D\end{psmallmatrix}\in SL_2(\Z)$ that reduces modulo $p$ to $X$.
Let $A_0= \begin{psmallmatrix}
    1 & 1 \\0 & 1
\end{psmallmatrix}, \ A_1 = \begin{psmallmatrix}
    1 & 0 \\1 & 1 \end{psmallmatrix}$. Consider the general case with generators $\tilde{A_0}=A_0^\alpha=\begin{psmallmatrix}
     1 & \alpha \\ 0 & 1
 \end{psmallmatrix}$, $\tilde{A_1}=A_1^\beta=\begin{psmallmatrix}
     1 & 0 \\ \beta & 1
 \end{psmallmatrix}$.

We write \[Y=\begin{pmatrix}A & B \\ C & D\end{pmatrix}=\begin{pmatrix}a_i & b_i \\ c_i & d_i\end{pmatrix}A_0^{\alpha q_i}A_1^{\beta q_{i-1}}\ldots A_0^{\alpha q_2}A_1^{\beta q_1}\]
where $\begin{pmatrix}a_0 & b_0\\ c_0 & d_0\end{pmatrix}=\begin{pmatrix}A & B\\ C & D\end{pmatrix}$. Since the matrix $X$ is obtained as a product of $\tilde{A_0}$ and $\tilde{A_1}$, there exists an integer $n$ such that $\begin{pmatrix}a_n & b_n\\ c_n & d_n\end{pmatrix} = A_0^{\alpha p_n}$ or $\begin{pmatrix}a_n & b_n\\ c_n & d_n\end{pmatrix} = A_0^{\alpha p_n}=A_1^{\beta q_n}$. We wish to determine the values of the $q_i$, $1\leq i \leq n$ for such an $n$. We have, for $1\leq i\leq n$, \begin{align}\label{euc-eqns} 
a_{i}&=a_{i-1}-\beta b_{i-1}q_{i-1} \nonumber \\
b_{i}&=b_{i-1}-\alpha a_{i-1}q_{i} \nonumber \\
c_{i}& = c_{i-1}-\beta d_{i-1}q_{i-1} \nonumber \\
d_{i}&= d_{i-1}-\alpha c_{i-1}q_{i} 
\end{align} where $\begin{pmatrix}a_0 & b_0\\ c_0 & d_0\end{pmatrix}=\begin{pmatrix}A & B\\ C & D\end{pmatrix}$ are the only fixed values of the $a_{i}, b_i, c_i, d_i$'s  and $\alpha, \beta$ are treated as integers smaller than $p$. 


 \paragraph{Case $\alpha=\beta =1$} 
First assume that $\alpha=\beta=1$.  Since $AD-BC=1 $, $\gcd(A,B)=\gcd(C,D)=1$, and we have either $A>B$ or $D>C$. We consider the case $A>B$, the argument for other case is analogous. Let $q_1, \ldots, q_n$ be the quotients that appear when Euclidean algorithm is applied to $(A,B)$, i.e.\ \begin{align*}
    A&=Bq_1+r_1 \\
    B&=r_1q_2+r_2 \\
&\vdots \\
r_i &= r_{i+1}q_{i+2} + r_{i+2} \\
   & \vdots \\
r_{n-2}&=r_{n-1}q_n  =\gcd(A,B)q_n=q_n
\end{align*}

 Note that $Y A_1^{-q_1}A_0^{-q_2} = \begin{psmallmatrix}
    r_1 & r_2 \\C-Dq_1 & -q_2(C-Dq_1)+D
\end{psmallmatrix}$, and for $i$ odd, $Y A_1^{-q_1}A_0^{-q_0}\ldots A_1^{-q_i} = \begin{psmallmatrix}
    r_{i} & r_{i-1} \\ \star & \star
\end{psmallmatrix}$ and $Y A_1^{-q_1}A_0^{-q_0}\ldots A_1^{-q_i}A_0^{-q_{i+1}} = \begin{psmallmatrix}
    r_i & r_{i+1} \\ \star & \star
\end{psmallmatrix}$. If $n$ is even then $Y A_1^{-q_1}A_0^{-q_2}\ldots A_0^{-q_{n}}= \begin{psmallmatrix}
    1 & 0 \\ q & 1
\end{psmallmatrix}= A_1^{x}$, for some $q\in \Z$, ($q<C$) since $r_{n-1}=1$. In other words, we have $Y= A_0^{q_n}A_1^{q_{n-1}} \ldots A_0^{q_2}A_1^{q_{1}+q}$  where $q_1, \ldots, q_n, q$ are obtained when the Euclidean algorithm is applied to $(A,B)$. The practical application of this principle to produce collisions of reasonable length is demonstrated in \cite{Tillich1993GrouptheoreticHF}. This gives a polynomial time solution for collisions.

\paragraph{Case $\alpha\neq 1, \; \beta \neq 1$}
In the case with general $\alpha, \beta$, there is no clear protocol to solve the set of equations~\eqref{euc-eqns} for the $q_{i}$'s. For instance, we start with the assumption that $A>B$, but it may still happen that $A\leq \beta B$, and then there is no clear way to apply the Euclidean algorithm.



 \subsection{Extending messages for diagonal hashes over $\fp$}
In this subsection, we discuss the extension of messages to produce diagonal hashes in $SL_2(\fp)$. The below lemma demonstrates the production of a diagonal matrix starting with the hash of an arbitrary message.

For a matrix $C$, we denote by $C[i,j]$ the $(i,j)^{\text{th}}$ entry of $C$. Further, a message is denoted in its binary string representation $m=0^{n_1}1^{m_1}\ldots 0^{n_r}1^{m_r}$ where multiplication corresponds to concatenation, and for $b\in\{0,1\}$, the string $b^n$ denotesf the concatenation of $n$ consecutive $b$-bits. 

\begin{lemma}\label{diag-hash}
      Let $z$ be any message and $C:=H(z)\in SL_2(\fp)$ be its corresponding hash value. Assume that $a:=C[0,0]\neq 0$. Then there exist integers $m, n \in \{0,1,\ldots, p-1\}$ such that $C\cdot A_0^m\cdot A_1^n$ is a diagonal matrix and $(C\cdot A_0^m\cdot A_1^n)[0,0]=C[0,0]$.
\end{lemma} \begin{pf} Write $C=\begin{pmatrix}a & b \\ c & d \end{pmatrix}$. Set $m=-b/(a\alpha)$ and $n=-ac/\beta$. Then \begin{align*}
     C\cdot A_0^m\cdot A_1^n 
=\begin{pmatrix}a(1+mn\alpha\beta)+nb\beta & 0 \\ c(1+mn\alpha\beta) +nd\beta & mc\alpha+d \end{pmatrix}  \end{align*} 
 Now, $c(1+mn\alpha\beta) +nd\beta= c(1+bc)-acd = c-c(ad-bc)=0$, since $C\in SL_2(\fpk)$. 
 Further, $A(1+mn\alpha\beta)+nb\beta = a(1+bc)-abc=a$, and $mc\alpha+d = -bc/a+d=1/a$. Thus, $C\cdot A_0^m\cdot A_1^n = \begin{pmatrix}a & 0 \\ 0  & a^{-1} \end{pmatrix}$. \end{pf}

\subsubsection{Generating Collisions}
In \cite{hardeasy},  the authors describe a method which uses the representation of 1 and pre-images of multiple diagonal matrices to produce a collision with the hash $H()$ of the empty word. We show here that Lemma~\ref{diag-hash} can be used to generate collisions in an alternative fashion. Let $z$ be any message and $C:=H(z) \in SL_2(\fp)$ be its hashed value. Observe that for any integers $m$ and $n$, the matrix $D:= A_1^n \cdot C \cdot A_0^m $ satisfies $D[0,0]=C[0, 0]$. Pick arbitrary distinct integers $n_1, m_1, n_2, m_2$ and set $D_1:= A_1^{n_1} \cdot C \cdot A_0^{m_1} $and $D_2:= A_1^{n_2} \cdot C \cdot A_0^{m_2} $. Then clearly, $D_1[0,0]=C[0, 0]=D_2[0, 0]$. Now apply Lemma 1 and find integers $\tilde{m}_1, \tilde{m}_2, \tilde{n}_1, \tilde{n}_2$ such that $\tilde{D_1}:=D_1\cdot A_0^{\tilde{m}_1}\cdot A_1^{\tilde{n}_1}$ and $\tilde{D_2}=D_2\cdot A_0^{\tilde{m}_2}\cdot A_1^{\tilde{n}_2}$ are diagonal and satisfy $\tilde{D_1}[0, 0]=D_1[0, 0]=C[0, 0]$ and $\tilde{D_2}[0, 0]=D_2[0, 0]=C[0, 0]$. Since $\tilde{D_1}$ and $\tilde{D_2}$ are diagonal, have the same entry, and lie in $SL_2(\fp)$, they must be equal, i.e.\ $\tilde{D_1}=\tilde{D_2}$. Defining $z_1=1^{n_1} \cdot z \cdot 0^{m_1+\tilde{m}_1}\cdot 1^{\tilde{n_1}}$ and $z_2=1^{n_2} \cdot z \cdot 0^{m_2+\tilde{m}_2}\cdot 1^{\tilde{n_2}} $. Clearly, $z_1\neq z_2$ but $H(z_1)=H(z_2)$.

\begin{example} 
For $p=7919$, $\alpha=5698$, $\beta= 6497$, consider the message text
\begin{align*}
z=&0^{44}1^{41}0^{17}1^{49}0^{47}1^{17}0^{50}1^{31}0^{15}1^{10}0^{39}1^{12}0^{2}1^{0}0^{24}1^{41}0^{28}1^{23}0^{9}1^{0}0^{47}1^{23}0^{1}1^{30}0^{18}\\&1^{32}0^{24}1^{14}0^{0}1^{49}0^{19}1^{28}0^{24}1^{26}0^{26}1^{26}0^{11}1^{1}0^{17}1^{20}0^{38}1^{22}0^{12}1^{38}0^{8}1^{33}0^{39}1^{42}0^{47}1^{29}\\&0^{10}1^{41}0^{14}1^{45}0^{13}1^{40}0^{42}1^{13}0^{2}1^{6}0^{40}1^{31}0^{2}1^{27}0^{1}1^{7}0^{36}1^{19}0^{3}1^{25}0^{10}1^{27}0^{21}1^{2}0^{12}1^{23}\\&0^{36}1^{8}0^{25}1^{39}0^{36}1^{0}0^{19}1^{39}0^{37}1^{32}0^{14}1^{4}0^{3}1^{12}0^{16}1^{23}0^{49}1^{25}0^{23}1^{19}0^{46}1^{23}0^{36}1^{31}\end{align*}
We have, $H(z)= \begin{pmatrix}
4812 & 5537 \\
4987 & 1690
\end{pmatrix}$. Choose random numbers $m_1=18, n_1=30$, $m_2=35
$, $n_2=33$ and compute $\tilde{m_1}= 6208, \tilde{n_1}= 744$, $\tilde{m_2}= 6191$, $\tilde{n_2}=180$ using Lemma~\ref{diag-hash}. Then for $z_1=1^{30}\cdot z \cdot 0^{6226}1^{744}$ and $z_2=1^{33}\cdot z \cdot 0^{6226}1^{180}$ we have the collision $H(z_1)=H(z_2)=\begin{pmatrix}
4812  &  0\\
0 & 1542
\end{pmatrix}$.


\end{example}

It is clear that the above method provides a deterministic technique to produce arbitratrily many distinct non-trivial collisions of size $\mathcal{O}(p)$. 
For the collisions to be of practical length, we further require that the values of the exponents $m$ and $n$ are considerably small relative to $p$ so that the resulting messages are of reasonable length.

\paragraph{Experimental results}
We looked for such $m$ and $n$ by random computer search in the above examples, using randomly generated messages of reasonable length and checking for the condition $m=-b/(a\alpha)<\sqrt{p}$ and $n=-ac/\beta<\sqrt{p}$. Unfortunately, for larger values of $p$, experiments indicate a very low probability of finding such values. For $30-40$ digit primes, brute force could no longer find any such examples. 

The below proposition describes a more structured approach for extending messages of the form $T=A_0^rA_1^s$.

\begin{proposition}\label{prop-diag-hash}  Consider a bound $\delta$. If there exist integers $r<\delta$, $y>p-\delta$ such that $s:=(y^{-1}-r^{-1})(\alpha\beta)^{-1}<\delta$ (where inverses are taken mod $p$) and $s(1+rs\alpha\beta)>p-\delta$, then for the message $T=A_0^{r}A_1^s$ (of length $\leq 2\delta)$, there exist $m,n<\delta$ such that $H(T)A_0^m A_1^n$ is diagonal.
\end{proposition}
\begin{pf}
We have $H(T)=\begin{pmatrix}        1+rs\alpha\beta & r\alpha \\ s\beta & 1    \end{pmatrix}$. As per Lemma~\ref{diag-hash}, the condition for $H(T)A_0^m A_1^n$ diagonal is that $m=-r/(1+rs(\alpha\beta))$ and $n=-s(1+rs(\alpha\beta))$. Write $x=(1+rs(\alpha\beta))$. For $m,n<\delta$ we require $rx^{-1}>p-\delta$, and $sx>p-\delta$. We set $y:=rx^{-1}$, from this we derive $s=(y^{-1}-r^{-1})(\alpha\beta)^{-1}$. Clearly, if $y>p-\delta$ and $s(1+rs\alpha\beta)>p-\delta$ then $m=p-y<\delta$ and $n=p-sx<\delta$, and by assumption, $r,s<\delta$, so the messages $T=A_0^rA_1^s$ and $A_0^mA_1^n$ each have length $\leq 2\delta$ and $H(T)A_0^mA_1^n$ is a diagonal matrix. 
\end{pf}

This clearly yields an algorithm with time complexity $\mathcal{O}(\delta^2)$ to test if for given values of $\alpha$, $\beta$, a message of the form   can be extended to have a diagonal hash. Indeed, one needs only to test all integers $r<\delta$, $y>p-\delta$ for the two conditions given in Proposition~\ref{prop-diag-hash}.

\subsection{Extending messages for triangular hashes over $\fpk$}
We now turn our attention to the more general case of producing upper or lower triangular hashes. In this subsection, we will consider the general case of $k\geq 1$ for the generalized Z\'emor hash functions over $\fpk$, and study the feasibility of efficiently producing messages hashing to upper triangular messages. 

Transforming the hash values of a message into an upper 
 or lower triangular matrix leads to producing collisions due to the following observation, which was explained in Propositon~3 of \cite{hardeasy}. Suppose that one can produce $r$ distinct messages $z_i$ such that $H(z_i)=\begin{pmatrix}a_i & b_i \\ 0 & d_i \end{pmatrix}$ is upper triangular. Also let $e_i$ be exponents such that $\prod\limits_{i=1}^r a_i^{e_i} = 1$. In other words, $(e_1, \ldots, e_r)$ is a solution to the representation problem in $\mathbb{F}_{p^k}^*$, which in turn reduces to a discrete log problem, since the multiplicative group of a finite field is cyclic. Then, for $z=z_1\ldots z_r $ we have $H(z)=\begin{pmatrix}1 & b \\ 0 & 1 \end{pmatrix}$ and thus $H(z||z) = \begin{pmatrix}1 & 0 \\ 0 & 1 \end{pmatrix}$ which collides with the hash value of the empty word. 

 We begin by proving some preliminary results below.

\begin{lemma} Suppose that the product $\alpha\cdot \beta\in \fp$ and that $C=\begin{pmatrix}a & b \\ c & d\end{pmatrix}$ is an arbitrary product of finitely many copies of $A_0$ and $A_1$. Then, $c/\beta \in \fp$ and $d \in \fp$.\end{lemma}
\begin{pf} Write $C=A_0^{m_1}A_1^{n_1}\ldots A_0^{m_r}A_1^{n_r}:=\begin{pmatrix} a_r & b_r \\
c_r & d_r\end{pmatrix}$. The proof proceeds by induction on $r$. For the case $r=1$, we have $C = \begin{pmatrix} 1+mn\alpha\beta & m\alpha \\
n\beta & 1\end{pmatrix}$. Clearly, $c_1/\beta = n_1 \in \fp$ and $d=1\in \fp$, so the result holds. 

Now assume that the result holds for $r\geq 1$. Write $m:=m_{r+1}$ and $n:=n_{r+1}$. We  have \[C_{r+1} = \begin{pmatrix} a_r & b_r \\
c_r & d_r\end{pmatrix}\cdot \begin{pmatrix} 1+mn\alpha\beta & m\alpha \\
n\beta & 1\end{pmatrix} = \begin{pmatrix} a_r(1+mn\alpha\beta) + nb_r\beta & a_rm\alpha +b_r\\
c_r(1+mn\alpha\beta) + nd_r\beta & c_rm\alpha +d_r\end{pmatrix}\] Now, \begin{align*}c_{r+1}^p/\beta^p =& 1/\beta^p[c_r^p(1+mn\alpha\beta)^p+nd_r^p\beta^p] \\
=& 1/\beta^p [c_r^p(1+mn\alpha\beta) + nd_r^p\beta^p] \\ =& (c_r/\beta)^p(1+mn\alpha\beta)+ nd_r^p.\end{align*} By the induction hypothesis we have $d_r^p=d_r\in \fp$ and $(c_r/\beta)^p=c_r/\beta \in \fp$, and so $c_{r+1}^p/\beta^p = (c_r/\beta)(1+mn\alpha\beta)+ nd_r = c_{r+1}/\beta \in \fp$. Further, \begin{align*} d_{r+1}^p = (c_rm\alpha+d_r)^p =& m c_r^p \alpha^p + d_r^p =m c_r^p \alpha^p + d_r = m c_r \cdot c_r^{p-1} \alpha^p + d_r \\ =& m c_r \cdot \beta^{p-1} \alpha^p + d_r \\ =&  m c_r \cdot (\alpha\beta)^{p-1} \alpha + d_r \\ =& mc_r\alpha+d_r = d_{r+1} \end{align*}
Hence, we have $c_{r+1}/\beta \in \fp$ and $d_{r+1}\in \fp$. By induction, the statement holds for all $r\geq 1$ and hence the lemma holds. 
\end{pf}

\begin{lemma}\label{triang-hash} Let $k\geq 1$ and $\alpha\cdot \beta\in \fp$. Let $z$ be any message and $C:=H(z)$ be its corresponding hash value. Assume that $a:=C[0, 0]\neq 0$. Then, there exist integers $m, n \in \{0,1,\ldots, p-1\}$ such that $C\cdot A_0^m\cdot A_1^n$ is upper triangular.
\end{lemma}
\begin{pf}
     Write $C=\begin{pmatrix}a & b \\ c & d \end{pmatrix}$. For integers $m$ and $n$, we have \begin{align*}
     C\cdot A_0^m\cdot A_1^n=&\begin{pmatrix}a & b \\ c & d \end{pmatrix}\begin{pmatrix}1+mn\alpha\beta & m\alpha \\ n\beta & 1 \end{pmatrix} \\
     &= \begin{pmatrix}a(1+mn\alpha\beta)+nb\beta & ma\alpha+b \\ c(1+mn\alpha\beta) +nd\beta & mc\alpha+d \end{pmatrix}
 \end{align*} 
Choose $m$ to be any integer such that $mc\alpha+d \neq 0$. Since $\alpha\beta\in \fp$, $c/\beta \in \fp$ and $d\in \fp$, we also have $c\alpha=\left(c/\beta\right)\cdot \alpha\cdot \beta \in \fp$, and we can define \begin{equation}\label{nval} n:=-\dfrac{c}{\beta(mc\alpha+d)} \in \fp\end{equation} This gives $C\cdot A_0^m \cdot A_1^n= \begin{pmatrix}a(1+mn\alpha\beta)+nb\beta & ma\alpha+b \\ 0 & mc\alpha+d \end{pmatrix}$.
\end{pf}


\begin{corollary}\label{cor1}
Let $k\geq 1$ and $\alpha\cdot \beta\in \fp$.    Let $\delta$ be a bound. Let $C=\begin{pmatrix} a & b \\ c & d \end{pmatrix}$ be the hash of an arbitrary message and $m\in \{0,1,\ldots, p-1\}$ be an integer such that both $m$ and $n=-c/(\beta(mc\alpha+d))\in \fp$ are smaller than $\delta$. Then $CA_0^mA_1^n$ has length at most $2\delta$ more than $C$ and is upper triangular.
\end{corollary}

\paragraph{Experimental results} Again, we looked for such $m$ and $n$ by random computer search in the above examples, using randomly generated messages of reasonable length and checking for the condition above. Once again, for larger values of $p$, experiments indicate a very low probability of finding such values. For $30-40$ digit primes, brute force could no longer find any such examples. 

Below, we explore a more structured approach for messages of the form $T=A_0^rA_1^s$.

\begin{lemma} Let $k\geq 1$ and $\alpha\cdot \beta\in \fp$. Consider a bound $\delta$. If there exist integers $m,s<\delta$ such that $s(ms\alpha\beta+1)^{-1}>p-\delta$, then for the message $T=A_0^{r}A_1^s$ (of length $\leq 2\delta$, here $r$ can be chosen freely), for $n=p-s(ms\alpha\beta+1)^{-1}<\delta$, $H(T)A_0^m A_1^n$ is upper triangular.
\end{lemma}
\begin{pf}
We have $H(T)=\begin{pmatrix}        1+rs\alpha\beta & r\alpha \\ s\beta & 1    \end{pmatrix}$. Let $m,s<\delta$ be integers such that $s(ms\alpha\beta+1)^{-1}>p-\delta$.  The condition for \small $H(T)A_0^m A_1^n=\begin{pmatrix}1+rs\alpha\beta & r\alpha \\ s\beta & 1 \end{pmatrix}\begin{pmatrix}1+mn\alpha\beta & m\alpha \\ n\beta & 1 \end{pmatrix}=\begin{pmatrix}(1+rs\alpha\beta)(1+mn\alpha\beta) +rn\alpha\beta& r\alpha \\ s\beta(1+mn\alpha\beta)+n\beta & 1+ms\alpha\beta \end{pmatrix}$ \normalsize being upper triangular is that $n=-s(1+ms(\alpha\beta))^{-1}$. The result is now clear.
\end{pf}




This clearly yields an algorithm with time complexity $\mathcal{O}(\delta^2)$ to test if for given values of $\alpha$, $\beta$, a message of the form $T=A_0^rA_1^s$ can be extended to one with an upper triangular hash, for a fixed bound $\delta$. Indeed, one needs only to test all integers $m,s<\delta$ for the two conditions given above.

Clearly, for $C\cdot A_0^m \cdot A_1^n$ to be upper triangular, $n$ needs to assume the value in equation~\eqref{nval}. Below, we derive the general condition (without the assumption $\alpha\cdot\beta\in \fp)$ on $m$ so that both the $m$ and $n$ lie in $\fp$.

\begin{proposition}If $\alpha\cdot\beta\not\in\fp$, then $C \cdot A_0^m \cdot A_1^n$ is upper triangular for $m,n\in \fp$ if and only if for \begin{equation} \label{gammaval}\gamma=\left(\dfrac{d ((d\beta)^{p-1}-c^{p-1})}{\alpha c^p(1-(\alpha\beta)^{p-1})}\right),\end{equation} we have $\gamma^p=\gamma$, and $m=\gamma; \ n=\dfrac{-c}{\beta(mc\alpha+d)}$.
\end{proposition}
\begin{pf}
     From the proof of Lemma~\ref{triang-hash}, it is clear that for $CA_0^mA_1^n$ to be upper triangular, Equation~\eqref{nval} must hold, along with $m^p=m$ and $n^p=n$. We have, \begin{align*}
   & -n = c/(\beta(m\alpha c+d)) \in \fp \\
     \implies & c^p/(\beta^p(m\alpha^pc^p+d^p)) = c/(\beta(m\alpha c+d))\\
     \implies & (c/\beta)^{p-1}\cdot(m\alpha c+d)=m\alpha^p+d^p \\
     \implies & c^{p-1}\cdot (m\alpha c+d) = \beta^{p-1}\cdot(m\alpha^p c^p+d^p)\\
     \implies & m\alpha c^p\cdot(1-(\alpha\beta)^{p-1})=\beta^{p-1}d^p-c^{p-1}d \\  \implies & m=\dfrac{   d ((d\beta)^{p-1}-c^{p-1})}{\alpha c^p(1-(\alpha\beta)^{p-1})}       \end{align*}  
Thus, for $m=\gamma$, one necessarily has $n:=\dfrac{-c}{\beta(mc\alpha+d)} \in \fp$. So, in order to obtain an upper triangular matrix, the only condition that needs to be satisfied is $\gamma^p=\gamma$.
\end{pf}

\begin{lemma}[\bfseries Case $k=2$]
Let $k=2$ and $\alpha\cdot \beta\not\in\fp$. As before, let $C=\begin{pmatrix}a & b \\ c& d\end{pmatrix}\in SL_2(\f_{p^{2}})$ be an arbitrary product of finitely many copies of $A_0$ and $A_1$. Then with $\gamma$ defined as in \eqref{gammaval}, $\gamma^p=\gamma$ always holds.  \end{lemma}

\begin{pf}
In this case, for any nonzero field element $y$ we have $y^{p^2-p}=y^{1-p}$, so \small\begin{align*}\gamma^p= \left(\dfrac{d ((d\beta)^{p-1}-c^{p-1})}{\alpha c^p(1-(\alpha\beta)^{p-1})}\right)^p & = (d^pc/\alpha^p)  \left(\dfrac{ ((d\beta)^{1-p}-c^{1-p})}{(1-(\alpha\beta)^{1-p})}\right) \\ &=  (d^p/c\alpha) 1/(d\beta c)^{p-1}  (\alpha\beta)^{p-1} \left(\dfrac{((d\beta)^{p-1}-c^{p-1})}{(1-(\alpha\beta)^{p-1})}\right) \\ &= d/(\alpha c^p) \left(\dfrac{((d\beta)^{p-1}-c^{p-1})}{(1-(\alpha\beta)^{p-1})}\right) =\gamma. \end{align*}\end{pf}
\normalsize

Thus, in $SL_2(\f_{p^{2}})$ one can always find collisions by right-multiplying an arbitrary matrix by a product $A_0^mA_1^n$. In other words, a message in $SL_2(\F_{p^2})$ is always extendable to have a triangular hash.

\begin{example} Consider $SL_2(\f_{p^{2}})$ for $p=239$, with generator denoted by $z$. Consider the following message text \begin{align*}
T=&0^{17}1^{25}0^{11}1^{8}0^{38}1^{33}0^{29}1^{27}0^{40}1^{24}0^{4}1^{2}0^{38}1^{16}0^{16}1^{39}0^{50}1^{42}0^{36}1^{22}0^{2}1^{41}0^{27}1^{29}0^{11}\\&1^{15}0^{29}1^{47}0^{29}1^{33}0^{47}1^{3}0^{35}1^{32}0^{49}1^{27}0^{24}1^{0}0^{21}1^{49}0^{33}1^{4}0^{50}1^{44}0^{42}1^{43}0^{4}1^{29}0^{14}1^{39}\\&0^{14}1^{15}0^{41}1^{0}0^{41}1^{1}0^{20}1^{34}0^{4}1^{6}0^{43}1^{5}0^{11}1^{7}0^{37}1^{29}0^{40}1^{20}0^{1}1^{13}0^{2}1^{14}0^{31}1^{41}0^{19}1^{24}\\&0^{50}1^{50}0^{25}1^{23}0^{30}1^{39}0^{6}1^{46}0^{39}1^{27}0^{8}1^{9}0^{25}1^{38}0^{0}1^{46}0^{15}1^{33}0^{47}1^{40}0^{40}1^{26}0^{48}1^{45}\end{align*}
 For a random choice of $\alpha$ and $\beta$, we may obtain hash $H(T)=\begin{pmatrix}134z + 110 & 131z + 185 \\ 74z + 17  & 58z + 41\end{pmatrix}$. Calculating $m$ and $n$ as per Lemma~\ref{triang-hash}, we see that $H(T)0^1 1^{10} = \begin{pmatrix}106z + 192  & 25z + 30 \\0 & 218z + 62\end{pmatrix}$.
\end{example}


A natural question arises here: can we generalize this method to make \[C\cdot A_0^{m_1}A_1^{n_1}\ldots A_0^{m_r} A_1^{n_r}\] upper/lower triangular and thereby extend the result to all $SL_2(\fpk)$? Since the major constraint is the condition $\gamma^p=\gamma$ (and the size of $\gamma$ does not seem easy to control), we first explore how the value of $\gamma$ changes when $C$ is multiplied on the right by a random product $A_0^mA_1^n$.

\begin{theorem}\label{gammas}
    Let $C=\begin{pmatrix}a & b \\ c& d\end{pmatrix}$
and $C':=C\cdot A_0^m \cdot A_1^n=\begin{pmatrix}a' & b' \\ c'& d'\end{pmatrix}$. Let $\gamma$ and $\gamma'$ be defined as in equation \eqref{gammaval}. Then, we have \[\gamma'=(c/c')^{p+1} (\gamma-m)\]
\end{theorem}
\begin{pf}   Let \[  X:=\dfrac{1}{\alpha(1-(\alpha\beta)^{p-1})}, \
  \gamma:=\dfrac{Xd}{c^p}[(d\beta)^{p-1}-c^{p-1}]
\]    We have, \begin{align*} \gamma' =& X \dfrac{d' ((d'\beta)^{p-1} - (n\beta d'+c)^{p-1}))}{(c')^p} = Xd' \dfrac{(d'\beta)^{p-1}-\frac{n\beta (d')^p+c^p}{n\beta d'+c}}{(c')^p} \\
    =& Xd' \dfrac{(d'\beta)^{p-1}(n\beta d+c))-n\beta (d')^p-c^p}{(c')^{p+1}} \\
    =& Xd'c \dfrac{(d'\beta)^{p-1}-c^{p-1}}{(c')^{p+1}}= Xd'c \dfrac{(d'\beta)^{p-1}-c^{p-1}}{(c')^{p+1}} \\
   =& \dfrac{Xcd'}{c'^{p+1} d'\beta} [(mc^p\alpha^p)\beta^p-c^{p-a1}(cm\alpha+d)\beta] \\
=& \dfrac{Xc}{{c'}^{p+1} \beta} [mc^p\alpha\beta((\alpha\beta)^{p-1}-1)+d\beta((d\beta)^{p-1}-c^{p-1})] \end{align*}
     Plugging the values of $X$ and $\gamma$ from above into the equation, we get \begin{align*}
         \gamma'=&\dfrac{Xc}{(c')^{p+1}\beta}mc^p\alpha\beta((\alpha\beta)^{p-1}-1) + \dfrac{Xcd}{(c')^{p+1}}\left(\dfrac{c^p\gamma}{Xd}\right) \\
         =& \left(\dfrac{c}{c'}\right)^{p+1}[-m +\gamma]
     \end{align*}
\end{pf}

For an extension where multiplication by a product $A_0^mA_1^n$ is allowed twice, we have the following result.

\begin{lemma}\label{triang-mul}
    For $C:=\begin{pmatrix}a & b \\ c& d\end{pmatrix}$, there exits integers $m_1,m_2,n_1,n_2$ such that $CA^{m_1}B^{n_1}A^{m_2}B^{n_2}$ is upper triangular if and only if the equation \begin{equation}\label{grob-pol}
        q_3 x^2 y+ q_2 xy + q_1 y + q_0=0 \end{equation}  has a solution $(x,y)\in \fp\times \fp$, where $q_0, q_1, q_2, q_3$ are given by \begin{align}\label{qs1} \begin{split}
q_3 &=c^{p^{2}}\alpha\beta((\alpha\beta)^{p^{2}-1}-1), \\
  q_2 &=c^{p^{2}}\gamma\alpha\beta(\gamma^{p-1}-(\alpha\beta)^{p^{2}-1}) + d\beta((d\beta)^{p^{2}-1}-1),\\
   q_1 &=d\beta \gamma (c^{p^{2}}\gamma^{p-1}-(d\beta)^{p^{2}-1}),\\
  q_0 &=c^{p^{2}}\gamma(\gamma^{p-1}-1).
  \end{split}
\end{align}
\end{lemma}
\begin{pf}
 Let $C':=C\cdot A^m \cdot B^n=\begin{pmatrix}a' & b' \\ c'& d'\end{pmatrix}$ and let $\gamma'$ be defined as in equation \eqref{gammaval}. By Theorem~\ref{gammas} we have $\gamma':= (c/c')^{p+1}(\gamma-m) \in \fp$. We treat $x:=m$ and $y:=n$ as variables in $\fp$. For $C'$ to be upper triangular for some values of $x$ and $y$, we require
    \begin{align*}
       & (\gamma')^p=\gamma' \\
      \implies & c^{p^2+p}(c')^{p+1}(\gamma^p-x) = c^{p+1}(c')^{p^{2}+p} (\gamma-x) \\
      \implies & c^{p^{2}}c'(\gamma^p-x) = c(c')^{p^2}(\gamma-x) \\
      \implies & c^{p^2}(\gamma^p-x)(c+(c\alpha x+d)y\beta) = c(c^{p^{2}})+ (xc^{p^2}\alpha^{p^2}+d^{p^2})y\beta^{p^2})(\gamma-x) \\      
\implies & c^{p^2-1}[\gamma^pc + \gamma^p c \alpha \beta xy+\gamma^pd\beta y -cx -c\alpha\beta x^{2y - d\beta xy}= c^{p^2}\gamma + (c\alpha\beta)^{p^2}\gamma xy+\\ & \qquad \qquad  \qquad \qquad \qquad \qquad  \qquad (d\beta)^{p^2}\gamma y- c^{p^2}x - (c\alpha\beta)^{p^2}x^2y - (d\beta)^{p^2} xy \\  
\implies & ((c\alpha\beta)^{p^2}-c^{p^2}\alpha\beta)x^2y + (-(c\alpha\beta)^{p^2}\gamma + \gamma^p c^{p^2}\alpha\beta+(d\beta)^{p^2}-d\beta) xy +\\ &  \qquad \qquad \qquad \qquad  \qquad \qquad \qquad (c^{p^2}\gamma^p(d\beta) - \gamma (d\beta)^{p^2})y + (\gamma^p c^{p^2}-\gamma c^{p^2})=0 \\ 
\implies & q_3 x^2 y+ q_2 xy + q_1 y + q_0=0
    \end{align*}
where $q_0, q_1, q_2, q_3$ are given as in Equation~\eqref{qs1}. The proof is now complete.
\end{pf}

  Note that for $q_0, q_1, q_2, q_3$ as in Equation~\eqref{qs1}, we can rephrase the condition in Lemma~\ref{triang-mul} as follows. $C:=\begin{pmatrix}a & b \\ c& d\end{pmatrix}$ is upper triangularizable by right multiplication if and only if the system of equations \begin{equation}\begin{split}
        & q_3 x^2 y+ q_2 xy + q_1 y + q_0=0 \\
        & x^p=x \\
        & y^p=y     
    \end{split}
    \end{equation} has a simultaneous solution $(x,y)$ in $\fpk\times \fpk$. One may then use Gr\"obner bases to solve this system of equations.
A corresponding conditions for $C$ to be lower triangularizable by left/right multiplication and to be upper triangularizable by left multiplication can be derived similarly. 

\begin{example}
    For simplicity, consider the field $\F_{2^{5}}$ with generator $z_5$ and $\alpha=z_5^3 + 1$, $\beta=z_5^3 + z_5^2 + 1$. Consider the hash matrix \\ \[C=\begin{pmatrix}
z_5^4 + z_5^3 + z_5^2 + z_5 & z_5^4 + z_5^3 + z_5^2 + z_5 \\        
z_5^3  &    z_5^4 + z_5^3 + z_5^2\end{pmatrix}.\] Here, we have $\gamma=z_5^4 + z_5 + 1
$ and the polynomial in Equation~\eqref{grob-pol} is $(z_5^2 + z_5)x^2y + (z_5^3 + z_5^2 + 1)xy + (z_5^3)y + (z_5^4 + z_5^2 + z_5)$. The $\langle (z_5^2 + z_5)x^2y + (z_5^3 + z_5^2 + 1)xy + z_5^3y + (z_5^4 + z_5^2 + z_5), x^p-x, y^p-y\rangle$ is trivial, so its Gr\"obner basis is $\{1\}$.
\end{example}

We attempted the above experiment with various different values of $p$ and $k$ and random matrices $C$ which are products of $A$ and $B$. In each case, without exception, we found no solution (i.e.\ the resulting Gr\"obner base of the ideal $\langle q_3x^2y+q_2xy+q_1y+q_0,x^p-x, y^p-y \rangle $ was $\langle 1 \rangle $). Therefore, this method of extension to produce a triangular matrix is not practically feasible.



\section{Generalized Tillich-Z\'emor Hash Functions}\label{gen-TZ}

In this section, we consider the generalized Tillich-Z\'emor hash function $\phi$ with the generators $A_0=\begin{pmatrix}\alpha & 1 \\ 1 & 0 \end{pmatrix}$ and $A_1=\begin{pmatrix}\beta & 1 \\ 1 & 0 \end{pmatrix}$ where $\alpha, \beta \in \fpk$.

 
More generally, treating $x$ as a variable, we consider the matrix $Y=\begin{pmatrix}x & 1 \\ 1 & 0 \end{pmatrix}$ and first compute its powers. Clearly, we have \begin{align*}
    Y^2=\begin{pmatrix}x^2+1 & x \\ x & 1 \end{pmatrix}, \ Y^3=\begin{pmatrix}x^3+x^2+x+1 & x^2+1 \\ x^2+1 & x \end{pmatrix}
\end{align*}
More generally, we can write 
\begin{align}\label{yn}
    Y^n=\begin{pmatrix}f_n(x) & f_{n-1}(x) \\ f_{n-1}(x) & f_{n-2}(x) \end{pmatrix}, \ n\geq 2
\end{align}
where $f_0(x)= 0$, $f_1(x)=1$, and \begin{equation}\label{recur}
f_n(x)=xf_{n-1}(x) + f_{n-2}(x)
\end{equation}
It is clear that the recurrence relation \eqref{recur} fully describes the powers of the matrix $Y$. The above formulation and recurrence relation were also derived in  \cite{onsec}.

\subsection{Computing $f_n(x)$ for characteristic $p\neq 2$}
In  \cite{onsec}, an expression for $f_n(x)$ was derived in the case $p=2$. The authors also calculate probabilities that the orders of A and B lie within certain bounds. In this section, we address the case $p\neq 2$ and derive a closed formula for $f_n(x)$. 

We may solve the recurrence Equation~\eqref{recur} by finding roots of the auxiliary polynomial $t^2-xt-1=0$. The quadratic formula then gives us \[t = \dfrac{x\pm \sqrt{x^2+4}}{2}.\]
Setting $r=\dfrac{x+ \sqrt{x^2+4}}{2}$ and $s=\dfrac{x- \sqrt{x^2+4}}{2}$, a general solution to \eqref{recur} has to be of the form $f_n(x)=ur^n+vs^n$. Note that we also have the initial conditions $f_1(x)=1$ and $f_0(x)=1$ (or equivalently $f_2(x)=x^2+1$). Plugging these in, we may solve for $u$ and $v$ to get \begin{align*}
    u=\dfrac{x+\sqrt{x^2+4}}{2\sqrt{x^2+4}}, \
    v= 1-u=\dfrac{-x+\sqrt{x^2+4}}{2\sqrt{x^2+4}}.
\end{align*} 
Write $w=\sqrt{x^2+4}$. Note that for any $n\geq 1$, we have \begin{align}\label{fnx}
    \begin{split}
    f_n(x)&=ur^n+vs^n = \left(\dfrac{x+w}{2w}\right)\left(\dfrac{(x+w)^n}{2}\right) + \left(\dfrac{-x+w}{2w}\right)\left(\dfrac{(x-w)^n}{2}\right) \\
    & = \left(\dfrac{(x+w)^{n+1}}{2^{n+1}w}\right)-\left(\dfrac{(x-w)^{n+1}}{2^{n+1}w}\right) \\
    & = \dfrac{1}{2^{n+1}w}\left[\sum\limits_{i=0}^{n+1}
    \binom{n+1}{i} x^i w^{n+1-i}- \sum\limits_{i=0}^{n+1} \binom{n+1}{i} x^iw^{n+1-i}\right] \\
    & = \dfrac{1}{2^{n+1}} \left[\sum\limits_{0\leq i \leq n+1 , \ n-i \text{ is even}} \binom{n+1}{i} x^i w^{n-i} \right] \\ &= \dfrac{1}{2^{n+1}}\left[\sum\limits_{0\leq i \leq n+1, \ n-i \text{ is even} }\binom{n+1}{i} \sum\limits_{j=0}^{(n-i)/2} \binom{(n-i)/2}{j}x^{i+2j}4^{(n-i)/2-j} \right] \\
    &= \dfrac{1}{2^{n+1}}\left[\sum\limits_{0\leq i \leq n, \ n-i \text{ is even}}\sum\limits_{j=0}^{(n-i)/2}
    \binom{n+1}{i}\binom{(n-i)/2}{j}2^{n-2j}x^{i+2j} \right]    \\
    \end{split}
\end{align}
Thus, $f_n(x)$ is always a polynomial in $\fp[x]$, and we have a closed formula for calculating it. Powers of $A_0$ and $A_1$ may therefore be computed in constant time.

Note that $\fpk$ is typically viewed through the isomorphism $\fpk \cong \fp[x]/\langle p(x) \rangle$ where $p(x)$ is an irreducible polynomial of degree $k$ over $\fp$. Thus, an element $\gamma \in \fpk$ is a polynomial of degree smaller than $k$, say $\gamma=g_\gamma(x)$. $f_n(\gamma)$ can then be calculated as a polynomial modulo $p(x)$ by simply composing $f_n$ and $g$, i.e.\ $f_n(\gamma)=f_n(g_\gamma(x)) \pmod{p(x)}$. 



We also have the following straightforward observation. A similar statement is found in \cite{onsec}.

\begin{lemma}
    Let $q(x)$ be any irreducible polynomial of degree $d$. Then, the sequence $\{f_n(x) \pmod{p(x)}\}$ is periodic, and its order divides $(p^{2k}-p^k)(p^{2k}-1)/(p^k-1)$.
\end{lemma}
\begin{proof}
    Clearly, for $Y=\begin{pmatrix}x & 1 
 \\1 & 0\\\end{pmatrix}$, Equation~\eqref{yn} gives us a formula for $Y^n$ in terms of $f_n(x)$, $f_{n-1}(x)$ and $f_{n-2}(x)$. Since $Y$ must have a finite multiplicative order $n_y$ in the finite group $SL_2(\fp[x]/\langle q(x)\rangle)$, we must have $f_{n_y-1}(x)=0$, $f_n(x)=f_{n-2}(x)=1$. From this point onwards, the sequence repeats its terms, and therefore its period is given by $n_y$, which must divide the order of the matrix group. 
\end{proof}

\begin{lemma}
    Suppose that the adversary can compute integers $m$ and $n$ such that $f_{n-1}(g_\alpha(x))=f_{m-1}(g_\beta(x)) \pmod{p(x)} $ and $f_{n-2}(g_\alpha(x))=f_{m-2}(g_\beta(x)) \pmod{p(x)}$. Then, the adversary can compute a collision of size $\mathcal{O}(\max(m,n))$ for the Generalized Tillich-Z\"emor hash function $\phi$.
\end{lemma}\begin{pf}
    In this case, it is clear from \eqref{recur} that $f_{n}(g_\alpha(x))=f_{m}(g_\beta(x)) \pmod{p(x)} $, and so $A_0^n=A_1^m$, i.e.\ $\phi(0^n)=\phi(1^m)$.
\end{pf}
    Note that the above result includes the cases $m=0$ and $n=0$, i.e.\ finding collisions with the identity matrix. However, by looking at the expression \eqref{fnx} for $f_n(x)$, one sees that even for the simplest equation $f_n(x)=0 \pmod{p(x)}$, finding a solution for the value of $n$ is not straightforward, since in \eqref{fnx} $n$ occurs both as a polynomial term (in the binomial coefficients) and in the exponent of $2$. This problem seems to be a complex amalgamation of a generalized discrete logarithm with a polynomial rather than a single power. One may extend the above result to products with more terms.

\begin{lemma} Let $\fp[x]/\langle q(x) \rangle$ be a finite field.
    If an adversary can find integers $m$ and $n$ such that the following relations hold
    \begin{align*}
       & f_m(f_n(x))+f_{m-1}(f_{n-1}(x))=1 \pmod {q(x)} \\
       & f_m(f_{n-1}(x))+f_{m-1}(f_{n-2}(x))=0 \pmod {q(x)} \\
       & f_{m-1}(f_{n}(x))+f_{m-2}(f_{n-1}(x))=0 \pmod {q(x)} \\
       & f_{m-1}(f_{n-1}(x))+f_{m-2}(f_{n-2}(x))=1 \pmod {q(x)},
    \end{align*} then $H(0^m 1^n)=H()$ gives a collision with the hash $H()$ of the empty word.
\end{lemma}

Note that in both the above results, the problem potentially becomes even more difficult if one constrains the values of $m$ and $n$ to yield practical-sized collisions.

\subsection{Malicious paramaters}

As seen in the above discussion, the security of the generalized Tillich-Z\"emor hash functions relies on the choice of the irreducible polynomial through which the finite field is constructed. Thus, a malicious construction of this polynomial can lead to the designer of the hash function being able to easily compute collisions.

  In \cite{attacking}, the authors produce malicious polynomials in finite fields $\f_{2^m}$, reducing the problem of collision-finding in the matrix group to the search for an irreducible polynomial of degree $m$, such that $m + 1$ is a proper divisor of $q - 1$ or $q + 1$, where $q = 2^m$. Over the field so constructed, the matrix $A$ has a small order 263, and leads also to a non-trivial collision, which the authors exemplify for $SL_2(\F_{2^{131}})$. This attack fails if $q - 1$ and $q + 1$ are both primes or if the factorization of their $\gcd$ involves large primes, since the collisions yielded are then too long.

In \cite{onsec}, the authors also deal exclusively with the characteristic 2 case, and propose an algorithm to decide whether the given
irreducible polynomial leads to  a vulnerable system under the attack of \cite{attacking}, and propose a solution to fix the vulnerability.  In the following result, we describe a malicious design for the choice of the polynomial defining the finite field, for odd characteristic $p$.

\begin{theorem}
    If one can find $N$ such that $\gcd(f_N(x)-1, f_{N-1}(x))$ has an irreducible divisor $q(x)$ of degree $d$, one can find a collision of size $\mathcal{O}(N)$ for the hash function $\phi(x)$ over the finite field ${\fp[x]}/{\langle q(x)\rangle}$. On the other hand, given a fixed finite field ${\fp[x]}/{\langle q(x) \rangle}$, if one can find an integer $N$ such that $q(x)$ divides $\gcd(f_N(x)-1, f_{N-1}(x))$ then one can find collisions of size $\mathcal{O}(N)$ for $\phi$.
\end{theorem}
\begin{pf}
    If one can find $N$ such that $\gcd(f_{N}(x), f_{{N}-1}(x))$ has an irreducible divisor $q(x)$ of degree $d$, then $f_{N}(x)=1\pmod q(x)$, $f_{N+1}(x)=0 \pmod q(x)$, then the sequence $\{f_{n}(x) \pmod {\fp[x]}/{\langle q(x)\rangle }\}$ has a period dividing $n_y$, and the multiplicative order of $Y$ in $SL_2({\fp[x]}/{\langle q(x) \rangle })$ divides $N$. Thus, one has a collision of $H(0^{N})$ with the hash of the empty word.
\end{pf}

\begin{example}
    $q(x)=x^{12} + 2x^{10} + x^6 + 2x^4 + 2x^3 + x^2 + x + 1$ is an irreducible polynomial over $\f_3$ such that $\{f_n(x) \pmod q(x)\}$ has period $531440=3^{12}-1$. Thus, in $SL_2(\F_{3^{12}})$, we always have collisions $H(0^{531440})=H(1^{531440})=H()$.
\end{example}

\bibliographystyle{plain}
\bibliography{references}

\end{document}